\documentclass[aps,pra,showkeys,showpacs,twocolumn,groupedaddress]{revtex4}

\usepackage{amsmath,paralist,amsthm,comment,amssymb}
\usepackage{multirow}
\usepackage{tikz}
\usetikzlibrary{shapes,backgrounds}

\newcommand{\ket}[1]{|#1\rangle}

\newcommand{\mc}[1]{\mathcal{#1}}

\newcommand{\F}{\mathbb{F}}
\newcommand{\nix}[1]{}

\newcommand{\supp}{{\rm supp}}
\newcommand{\rk}[1]{{{\rm rk}\left( #1\right)}}

\newcommand{\be}{\begin{eqnarray*}}
\newcommand{\ee}{\end{eqnarray*}}
\newcommand{\ben}{\begin{eqnarray}}
\newcommand{\een}{\end{eqnarray}}

\newcommand{\ba}{\begin{array}}
\newcommand{\ea}{\end{array}}
\newcommand{\bmt}{\left[\begin{array}}
\newcommand{\emt}{\end{array}\right]}

\newtheorem{theorem}{Theorem}
\newtheorem{lemma}[theorem]{Lemma}
\newtheorem{corollary}[theorem]{Corollary}

\newtheorem{fact}{Fact}

\begin{document}


\title{Matroids and Quantum Secret Sharing Schemes}


\author{Pradeep Sarvepalli}
\email[]{pradeep@phas.ubc.ca}
\affiliation{Department of Physics and Astronomy,
University of British Columbia, Vancouver V6T 1Z1, Canada }
\author{Robert Raussendorf}
\affiliation{Department of Physics and Astronomy,
University of British Columbia, Vancouver V6T 1Z1, Canada }

\date{October 23, 2009}

\begin{abstract}
A secret sharing scheme is a cryptographic protocol to distribute a secret state in an encoded form among a group of players 
such that only authorized subsets of the players can reconstruct the secret.  Classically, efficient secret sharing schemes have been 
shown to be induced by matroids. Furthermore,  access structures of such schemes can be characterized by an excluded minor relation. No such relations are known for quantum secret sharing schemes. In this paper we take the first steps toward a matroidal characterization of quantum secret sharing schemes.  In addition to providing a new perspective on quantum secret sharing schemes, this characterization has important benefits.  While previous work has shown how to construct quantum secret sharing schemes for general access structures, these schemes are not claimed to be efficient. In this context the present results prove to be useful; they enable us to construct efficient quantum secret sharing schemes for many general access structures. More precisely, we show that an identically self-dual matroid that is representable over a finite field induces a pure state quantum secret sharing scheme with information rate one. 
\end{abstract}


\maketitle

\section{Introduction}\label{sec:intro}
Secret sharing is an important cryptographic primitive originally motivated by the need to distribute secure information among parties some of whom are untrustworthy \cite{shamir79,blakley79}.  Additionally, it finds   applications in  secure multi-party distributed computation \cite{cramer08,benor06}. Secret sharing schemes have a rich mathematical structure \cite{martin91} and they have been shown to be closely associated to error correcting codes \cite{massey93,cramer08,simonis98} and matroids \cite{brickell91,golic98,beimel08,seymour92, stinson92,cramer08}. The interplay with these objects has enabled  us to obtain new insights not only about secret sharing schemes but codes and matroids as well. 
Although relatively new, the field of quantum secret sharing \cite{hillery99} has made rapid progress both theoretically 
\cite{cleve99,karlsson99,gottesman00,imai04,xiao04,rietjens05,markham08} and experimentally \cite{tittel01, gaertner02,lance04,lance04b}. 
However, its connections with other mathematical disciplines have not been as well studied. 
In particular, no connections have been made with the theory of matroids, which is in sharp contrast to the classical scenario. 
These connections are of more than theoretical interest. Classically, optimal secret sharing schemes i.e. those with 
information rate one, are induced by matroids. Additionally, matroids provide alternate methods to prove bounds on the rates that can be achieved for certain access structures.  
For all these reasons it is useful to develop the theory of matroids and quantum secret sharing schemes.

In this paper it is our goal to bring into bearing the theory of matroids to characterize quantum secret sharing schemes. While our results are only the first steps toward this characterization, they do indicate the usefulness of such associations.
This paper is organized as follows.  We begin with a brief review of the necessary background in 
secret sharing. In Section~\ref{sec:mat}
we review some of the known results on classical secret sharing schemes and matroids; 
these results are not well known in the quantum information community and also provide the backdrop for generalizing the 
connections between matroids and secret sharing schemes. 
In Section~\ref{sec:main} we prove the central result of this paper, namely how representable identically self-dual matroids lead to
efficient quantum secret sharing schemes.
We assume that the reader is familiar with the basic results on quantum computing and stabilizer codes.

\subsection{Classical secret sharing}
A secret sharing scheme is a protocol to distribute a secret $s$ among a set of players $P$, by a dealer
$D$, such that only authorized subsets of $P$ can reconstruct the secret. 
Subsets of $P$ which cannot reconstruct the secret are called unauthorized
sets.  The access structure $\Gamma$ consists of all subsets that can reconstruct the secret. The adversary structure $\mc{A}$ consists of all unauthorized  subsets.  Any access structure $\Gamma$  is required to satisfy the monotone property i.e. if 
$A \in \Gamma$, then any set $B\supseteq A$ is also in $\Gamma$. This is the only restriction on the access structures for classical secret sharing schemes. Any access structure satisfying the monotone property can be realized by an appropriate secret sharing scheme albeit with 
large complexity, see for instance \cite{stinson92}. A secret sharing scheme is said to be perfect if the unauthorized sets cannot extract any information about the 
secret. A precise information theoretic formulation can be given that quantifies this condition. 
We typically require the secret to be taken from a finite alphabet, $\mathbb{S}$. The shares distributed need not be
in the same domain as the secret; in fact each share can be in a domain of different alphabet. Let the domain of the $i$th party be $\mathbb{S}_i$. 
An important metric of performance for secret sharing schemes is the information rate $\rho$ which is defined as 
\ben
\rho = \min_{i} \frac{\dim \mathbb{S}}{\dim \mathbb{S}_i}. \label{eq:infoRate}
\een
Secret sharing schemes with $\rho=1$ are said to be ideal. The associated access structure is 
said to be ideal. More generally if an access structure can be realized with information rate one
for some secret sharing scheme, then it is said to be ideal. Note that we do not restrict the 
dimension of the secret in this case. An important problem of secret sharing
is to construct ideal secret sharing schemes for any given (monotone) access structure. 
Not every access structure can be realized with information rate of one. 

\subsection{Quantum secret sharing}
A quantum secret sharing scheme generalizes the classical one in two possible ways. 
We use quantum states to share either a secret quantum state  or a classical secret. 
Some authors refer to the first case as quantum state sharing, reserving the term ``quantum secret sharing" to situations 
where the secret is shared in an adversarial setting. Though this  might be preferable in some contexts, we will continue to use the traditional terminology. 
Quantum secret sharing schemes for classical secrets were introduced by Hillery et al in  \cite{hillery99}. They also proposed
schemes for sharing quantum secrets, however these are not perfect i.e.,  unauthorized sets can extract 
some information about the secret. Cleve et al \cite{cleve99} proposed the first perfect quantum secret sharing schemes for quantum secrets.
The theory of quantum secret sharing was developed further making important connections to quantum coding theory in \cite{cleve99,gottesman00} and quantum information theory in \cite{imai04,rietjens05} and more recently to graphs via labelled graph states in \cite{markham08}. 

In this paper we are concerned with the sharing of quantum secrets. Unlike classical secret sharing schemes a quantum secret sharing scheme cannot realize every monotone access structure. An additional constraint due to the ``no-cloning theorem'' \cite{wootters82,dieks82} has to be imposed on a realizable access structure.  Recall that the no-cloning theorem states that an arbitrary quantum state cannot be copied.
In any quantum secret sharing scheme we cannot have two disjoint authorized sets in the access structure as this would
violate the no-cloning theorem. This condition in conjunction with the monotonicity of access structure determines the allowed
access structures for all quantum secret sharing schemes \cite[Theorem~8]{gottesman00}. The same condition has been stated in different forms in the literature.  We record this result in its various forms for future use. First we need the notion of dual of a set. 
Let $P$ be a set, then we denote the powerset of $P$ as $2^P$.  For any subset  $A \subseteq 2^P$,  we define the dual of $A$ as 
\ben
A^\ast= \{ x\subset P \mid \overline{x} \not\in A \}.\label{eq:dual}
\een

\begin{lemma}[Self-orthogonal access structures]\label{lm:selfOrth}
Let $\Gamma$ be the access structure and $\mc{A}$  the adversary structure of a  quantum secret sharing scheme. 
Then the following statements are equivalent. 
\ben
{A\cap B \neq \emptyset \mbox{ for all } A, B\in\Gamma} \label{eq:noDisjointSets} \\
{\Gamma \subseteq \Gamma^\ast\label{eq:selfOrth}}\\
{\mc{A}^\ast \subseteq \mc{A}\label{eq:dualContaining}}
\een
\end{lemma}
\begin{proof}
We shall show that (\ref{eq:noDisjointSets}) $\Rightarrow$ (\ref{eq:selfOrth}).  It follows that if $A \in \Gamma$, then 
$\overline{A}\not\in \Gamma$ 
as $A\cap \overline{A}=\emptyset$.  But $\Gamma^\ast= \{ B\mid \overline{B}\not\in \Gamma\}$. Since 
$\overline{A}\not\in \Gamma$ it follows that $A\in \Gamma^\ast$ and  $\Gamma\subseteq \Gamma^\ast$. Conversely, let $\Gamma\subseteq  \Gamma^\ast$. Then from the definition of 
$\Gamma^\ast$, it follows that for any $A\in \Gamma$, we must have $\overline{A} \not\in\Gamma$ i.e. $\overline{A} \in\mc{A}$.  
Further all subsets of $\overline{A}$ are also in $\mc{A}$. Now assume that there exists some $B\in \Gamma$ such that 
$A\cap B=\emptyset$. Then $B\subseteq \overline{A}$. But all subsets of $\overline{A} \in \mc{A}$ i.e. they are not in $\Gamma$
which contradicts that $B\in \Gamma$. Therefore there exists no subset $B\in \Gamma$ such that $A\cap B =\emptyset$ proving
that (\ref{eq:selfOrth}) $\Rightarrow$ (\ref{eq:noDisjointSets}).

Now we shall show that  (\ref{eq:selfOrth}) $\Leftrightarrow $ (\ref{eq:dualContaining}). Assume that 
(\ref{eq:selfOrth}) holds. Then since $\Gamma\cap \mc{A}=\emptyset$ and $\Gamma\cup \mc{A} = 2^P = \Gamma^\ast \cup \mc{A}^\ast$, we have that
$\mc{A}  = (\Gamma^\ast \cup \mc{A}^\ast)\setminus \Gamma  = (\Gamma^\ast \setminus \Gamma) \cup \mc{A}^\ast$, where
we used the fact that $\Gamma^\ast\cap \mc{A}^\ast=\emptyset$ and $\Gamma\subseteq \Gamma^\ast$. It now follows that $\mc{A}^\ast \subseteq \mc{A}$ and 
(\ref{eq:dualContaining}) holds. Now assume that (\ref{eq:dualContaining}) holds, then again we have 
$\Gamma\cup \mc{A} =  \Gamma^\ast \cup \mc{A}^\ast$ and this time we can write $\Gamma^\ast=( \Gamma\cup \mc{A} )\setminus \mc{A}^\ast = (\Gamma^\ast \setminus \Gamma) \cup \mc{A}^\ast$ and therefore $\Gamma^\ast \supseteq \Gamma$ and  
 (\ref{eq:selfOrth}) holds. 
\end{proof}

We often refer to an access structure that is realizable by a quantum secret sharing scheme as a quantum access structure.
Smith \cite[Theorem~1]{smith00}  characterized  the adversary structure of quantum secret sharing schemes as in (\ref{eq:dualContaining}).  
Condition \eqref{eq:selfOrth} is somewhat reminiscent of the requirement for self-orthogonal classical codes for
quantum error correction.  If $\Gamma =\Gamma^\ast$, then we say that the access structure is self-dual.

A quantum secret sharing scheme which encodes a pure state secret into a global pure state is said to be a pure state scheme and a mixed state scheme if it encodes into a global mixed state. Self dual access structures can be realized by pure state schemes, where as 
non-self-dual access structures can be realized only as mixed state schemes. 
 A theorem \cite[Theorem~3]{gottesman00} due to Gottesman  shows that every mixed state scheme can be derived from a pure state scheme. So we do not lose any generality by focussing on the pure state schemes. 
The simplest access structures are the $((k,n))$ threshold access structures---in this case, the authorized sets are any subset of
size $\geq k$ and unauthorized sets are subsets of cardinality less than $k$. 
 Smith \cite{smith00} and independently Gottesman \cite{gottesman00} showed how to construct quantum secret sharing schemes with general access structures. 

In studying general access structures it is often convenient to work with the minimal access structures, which are the generating sets of the access structures. We define the minimal access structure $\Gamma_m$ of the access structure $\Gamma$ as 
\ben
\quad \Gamma_m =\{A \in \Gamma \mid  B \not\subset A \mbox{ for any } B\in \Gamma \}. \label{eq:minAccessStruct}
\een
If every party in $P$ occurs in at least one minimal authorized set of $\Gamma$, then we say that the access structure is connected.
We restrict our attention to such access structures in this paper.
Our primary goal in this paper is to explore connections of quantum secret sharing  schemes with matroids and characterizing
the associated access structures in terms of matroids if it is possible. We also address the construction of secret sharing schemes.
Our constructions make use of  CSS codes reminiscent of the constructions of Smith for general access structures. 

\section{Matroids and Secret Sharing}\label{sec:mat}

Matroids have been associated to secret sharing schemes \cite{cramer08,brickell91}, also see \cite{stinson92}
for a brief overview of some of the main results. Such schemes which are induced by a matroid are called matroidal. 
Useful results with respect to characterization and performance of secret sharing schemes can be derived by means of such an association, \cite{brickell91,beimel08}. 
Also, such an association also implies an implicit correspondence between matroids and access structures. In fact, classically, most of
the associations focus on this correspondence and tend to ignore the scheme realizing the access structure.
By far we do same however, since a given access structure might not be a quantum access structure
we do bear in mind that we cannot entirely ignore the fact that the access structure is being realized through a quantum scheme.
It is important to note that not every secret sharing scheme can be associated to a matroid. 

\subsection{Matroids}
First we recall a few facts about matroids, readers interested in a comprehensive introduction to matroids can refer to \cite{oxley04}.

A set $V$ and $\mc{C} \subseteq 2^V$ form a matroid $\mc{M}(V, \mc{C})$ if and only if the following conditions hold. For any 
$A, B \in \mc{C}$ and $A\neq B $
\begin{compactenum}[M1)]
\item $A\not \subseteq B$.
\item If $x\in A\cap B$, then there exists a $C\in \mc{C}$ such that $C\subseteq (A\cup B)\setminus \{x\} $.
\end{compactenum}
We say that $V$ is the ground set  and $\mc{C}$ the set of circuits of the matroid. A proper subset of any circuit is
said to be independent while a set containing any circuit is said to be dependent. 
With every matroid we define a nonnegative integer valued function called the rank function $\text{rk}:V\rightarrow \mathbb{N}$ as
\ben
\rk{X} = |I|, \label{eq:rankfn}
\een
where $ I \subseteq X \subseteq V$  is a maximal independent subset of $ X$. 
A matroid is said to be (linearly) representable over a field
 $\mathbb{F}$ if the ground set can be identified with the columns of a matrix (over $\mathbb{F}$) and the circuits with the minimal dependent columns of the matrix. In this paper we are only interested in finite fields.
We can also define matroids in terms of their bases, which are maximal independent sets of $V$. 
A set $V$ and $\mc{B} \subseteq 2^V$ form a matroid $\mc{M}(V, \mc{B})$ if and only if the following conditions hold.
\begin{compactenum}[B1)]
\item $\mc{B} \neq  \emptyset$.
\item If $B_1,B_2 \in \mc{B}$ such that $x\in B_1\setminus B_2$, then there exists a $y\in B_2\setminus B_1$ such that $(B_1\setminus x) \cup \{ y\} \in \mc{B} $.
\end{compactenum}
Given a matroid $\mc{M}(V,\mc{B})$ we define its dual matroid $\mc{M}(V,\mc{B})^\ast$ as the matroid with ground set $V$ and bases $\mc{B}^\ast =\{V\setminus B\mid B\in \mc{B} \}$ i.e. $\mc{M}(V,\mc{B})^\ast= \mc{M}(V,\mc{B}^\ast)$.

\subsection{Secret sharing schemes from matroids}
Given a matroid $\mc{M}$ we can associate a secret sharing scheme to $\mc{M}(V,\mc{C})$. 
We assume that the ground set of the matroid is given by $V= \{0,1, \ldots,n-1, n \}$. 
We identify one of the elements of the ground set, say $i\in V$,  as the dealer and then list all the circuits of $\mc{M}$ that contain $i$. 
Let this be denoted as 
\ben
\Gamma_{i,m} = \{ C \mid C\cup i \in \mc{C} \}.\label{eq:mat2ssMin}
\een
Consider the access structure given by 
\ben
\Gamma_i = \{A\mid V \supseteq A \supseteq C \mbox{ for some } C\in \Gamma_{i,m} \}.\label{eq:mat2ss}
\een
We can easily verify that $\Gamma_i$ is a monotonic and that its minimal access structure is given
by $\Gamma_{i,m} $.
Since any monotonic access structure can be realized as a secret sharing scheme
every matroid defines an access structure. This result is stated in the following fact,  see \cite{cramer08}.
\begin{fact}\label{fc:mat2ss}
Every matroid $\mc{M}(V,\mc{C})$ induces an access structure $\Gamma_i$ as defined in equation
\eqref{eq:mat2ss}. 
\end{fact}

Please note that the above association is in a sense nonconstructive, it does not specify how to derive the associated secret sharing scheme; it merely states that there exists a secret sharing scheme that can realize the induced access structure $\Gamma_i$. Further, depending on which element of the ground set of the matroid is identified as the dealer,  we may obtain many schemes with
possibly different access structures from the same matroid. 

A natural question that we are faced with is how to make this association constructive and determine the bounds on the information 
rate of the resulting access structure. Brickell and Davenport \cite{brickell91} showed that if the matroid is representable over a 
finite field \footnote{Strictly, \cite[Theorem~2]{brickell91} only requires the matroid to be representable over a near field.}, then we obtain ideal secret sharing schemes and access structures.

However, if the matroid is not representable, then we can no longer be certain if the matroid induces an ideal secret sharing scheme. Seymour proved that there exist non-representable 
matroids which cannot induce an ideal secret sharing scheme \cite{seymour92}, while Simonis and Ashikhmin \cite{simonis98} showed that there exist 
non-representable  matroids, such as the non-Pappus matroid, which induce ideal schemes. However, these latter 
matroids---while not affording a linear representation---can be multilinearly represented. Matroids which induce 
ideal access structures are called ss-representable matroids \cite{padro07}.  They may not be linearly representable.

\subsection{Matroids from secret sharing schemes}
Given that we can obtain secret sharing schemes from matroids, we could ask if the converse is possible.  
As we mentioned earlier,
such a correspondence does not exist for all secret sharing schemes. We review some of the related work in this context. The 
correspondence between the matroids and secret sharing schemes naturally implies that the access structure is associated to the circuits 
of the matroid. This association could involve the scheme explicitly. However a result due to Martin \cite{martin91}, see also \cite{stinson92},  shows that we can associate the access structure to a matroid independently of the scheme used to realize that structure. This involves a function, say  $f$,  defined on the  space of access structures; $f$ maps an access structure
to an ordered pair, which may or may not be a matroid. If $f(\Gamma)$ is a matroid, then we say that $\Gamma$ is matroid-related.
The minimal access structure will play a more important role in this regard. As usual we denote by  $P$ the set of  participants 
and by $D$ the dealer. Define the extended access structure $\Gamma_e= \{ A\cup D\mid \mbox{ for all }  A\in \Gamma_m\}$.
Further let 
\ben
\mathbb{J}(A,B) & =& A\cup B \setminus \left( {\bigcap}_{C\in \Gamma_e : C\subseteq A\cup B} C \right)\\
\mc{C}_\Gamma &=& \left\{\ba{l} \mbox{ minimal sets of }\mathbb{J}(A,B) \mbox{ for} \\
\mbox{ all } A,B\in \Gamma_m \mbox{ and }A\neq B \ea\right\}.
\een
We let $f(\Gamma) = (P\cup D, \mc{C}_\Gamma)$. 
If $\mc{C}_\Gamma$ satisfies the axioms M1 and M2, then we associate $\Gamma$ to the  matroid $\mc{M}_\Gamma$ 
whose ground set is  $P\cup D$ and the set of circuits are given by 
$\mc{C}_\Gamma$ i.e.
\ben
\mc{M}_{\Gamma}= \mc{M}(P\cup D, \mc{C}_\Gamma).\label{eq:ss2Matroid}
\een
  
This definition of the matroid is in terms of the circuits that can be formed from the ground set. 
{\em We could always define a structure from the secret sharing scheme, equivalently its access structure,
as above but the resulting structure is not necessarily a
matroid}. It is a matroid only under certain conditions. Only when $(P\cup D, \mc{C}_\Gamma)$ induce a matroid we say that $\Gamma$ is matroid related.

Classically an access structure induces a matroid only when it satisfies certain conditions. Before we can state this condition precisely we need the notion of  minors.  Let $\Gamma$ be an access structure, then we define two operations of 
deletion and contraction, which we denote  by ``$\setminus$'' and ``$/$'' respectively.  Given a set $Z\subseteq P$  we define 
\ben
\Gamma\setminus Z &=& \{A \subseteq P\setminus Z  \mid A  \in \Gamma \},\label{eq:delete}\\
\Gamma / Z &=& \{ A \subseteq P\setminus Z \mid A\cup Z \in \Gamma \}.\label{eq:contract}
\een
An access structure $\Gamma'$ derived from $\Gamma$ through a sequence of deletions and contractions is called a
minor of $\Gamma$. 
A result by Seymour \cite{seymour76} shows that the access structures are matroid related if the access structure satisfies
a forbidden minor relation. 
\begin{lemma}[Seymour]\label{lm:matroidalAccess}
An access structure $\Gamma \subseteq 2^P$ is matroid related if and only if it  does not have the following 
minors:
\begin{compactenum}[a.]
\item $\Gamma_a= \{ \{1,2 \}, \{2,3 \}, \{3,4 \} \}$
\item $\Gamma_b=  \{1,2 \}, \{1,3 \}, \{1,4 \}, \{2,3 \} \}$
\item $\Gamma_c= \{ \{1,2 \}, \{1,3 \}, \{2,3,4 \} \}$
\item $\Gamma_d = \{ \{ 1,\ldots, s\}, \{1, s+1 \}, \ldots, \{s, s+1\} \}$
\end{compactenum}
where $P=\{ 1,\ldots, 4\}$ except in $d$ where $P=\{ 1,\ldots s, s+1\}$
and $s\geq 3$.
\end{lemma}
Please note that in the preceding result, the minimal access structures are given rather than the complete access structure.
Seymour originally stated this result in terms of matroid ports. The reformulation we have given here in terms of the access structures 
is due to  Mart\'{i}-Farr\'{e} and Padr\'{o} \cite{padro07}. This result together with Lemma~\ref{lm:selfOrth} immediately provides us with a criterion as to which quantum access structures can be induced by matroids.

Self-orthogonality, however, is not a property inherited by minors of access structures.   For instance contraction does not always
preserve the self-orthogonality of the access structures.  Consider the following (minimal) access structure:  $\Gamma = \{ \{1,2,3 \},  \{2,3,4 \}, \{3,4,5 \}\}$. Then 
$\Gamma/3 = \{ \{1,2 \}, \{2,4 \}, \{4,5 \} \}$.  In this case we have two disjoint authorized sets; such an access structure cannot be realized by a quantum secret sharing scheme as it would lead to a violation of the no cloning theorem. 
Therefore, it is not possible to determine a result similar to Lemma~\ref{lm:matroidalAccess} for self-orthogonal access structures i.e. 
the forbidden minors for access structures that are self-orthogonal.  Incidentally, there exist other important classes of matroids such as transversal matroids which are not minor closed. 

Brickell and Davenport \cite[Theorem~1]{brickell91} showed that every classical ideal access structure induces a matroid. 
In figures~\ref{fig:accessVenn}~and~\ref{fig:accessVennQ} we summarize the relation between permissible access structures,
matroidal access structures, and ideal access structures for classical schemes and quantum schemes. 
We do not know if every access structure that is realized by an ideal quantum secret sharing scheme induces a matroid. 
Therefore we show that  the set of ideal quantum access structures does not lie entirely in the set of matroidal access structures
in figure~\ref{fig:accessVennQ}.

\medskip
\begin{center}
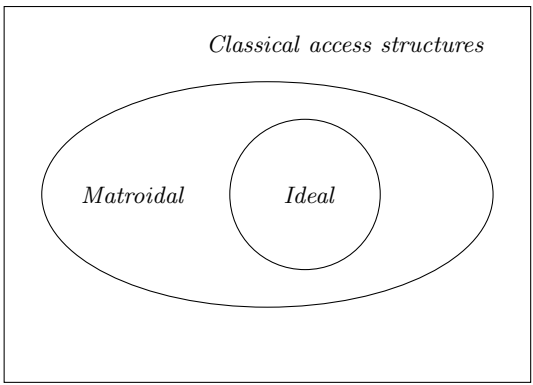
\begin{figure}[htb]
\begin{tikzpicture}
    \draw (-3.5,-2.5) rectangle (3.5,2.5); \draw (3,2) node[left] {\em Classical access structures};
    \draw (0,0) ellipse  (3cm and 1.5cm);  \draw (-1,0)  node[left] {\em Matroidal};
    \draw (0:0.5cm) circle (1cm);\draw (1,0) node [left] {\em Ideal}; 
\end{tikzpicture}
\caption{Relation between ideal, matroidal and general classical access structures}\label{fig:accessVenn}
\end{figure}
\end{center}

\begin{center}
\begin{figure}[htb]
\begin{tikzpicture}
   \draw (6.5,-2.5) rectangle (13.5,2.5); \draw (13,2) node[left] {\em Quantum access structures};
    \draw (10,0) ellipse  (3cm and 1.5cm);  \draw (9,0)  node[left] {$Matroidal$};
     \draw (10,-1.25cm) circle (1cm);     
   \begin{scope}
    \clip (6.5,-2.5) rectangle (13.5,2.5); 
    \clip (10,0) ellipse  (3cm and 1.5cm); 
     \fill[gray] (10,-1.25cm) circle (1cm) ;\draw (10.5,-1.2) node [left] {\em Ideal};
     \end{scope}
\end{tikzpicture}
\caption{Relation between ideal, matroidal and general quantum access structures. It is possible that all ideal quantum access structures are also matroidal.}\label{fig:accessVennQ}
\end{figure}
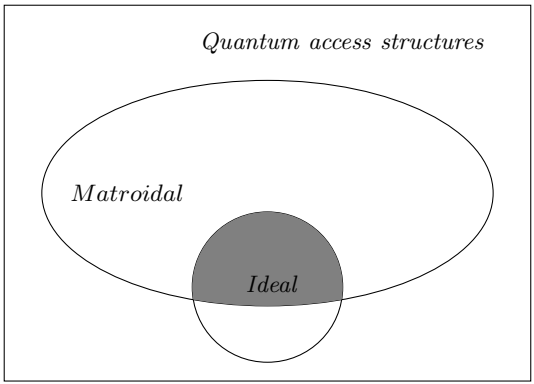
\end{center}

\section{Relating Matroids and Quantum Secret Sharing}\label{sec:main}
\subsection{Matroidal quantum secret sharing schemes}
In this section we present the central result of our paper, Theorem~\ref{th:isdMat2qss}. It shows that a class of matroids induce ideal pure state quantum secret
sharing schemes. First we need the following preliminaries. 
We say a  matroid is self-dual if it is isomorphic to its dual matroid. If it is equal to its dual matroid then we say it is 
an identically self-dual (ISD) matroid. 
\begin{fact}\label{fc:dualMat2ss}
Let $\Gamma_i$ and $\Gamma_i^d$ be the access structures induced by a matroid $\mc{M}(V,\mc{C})$ and its dual matroid 
$\mc{M}^\ast$  by treating the $i$th element as the dealer. Then  we have 
\ben
\Gamma_i^d = \Gamma_i^\ast \label{eq:dualMat2ss}
\een
\end{fact}
Fact~\ref{fc:dualMat2ss} was stated in \cite{cramer08}. Together with Lemma~\ref{lm:selfOrth}, 
and the fact that every self-dual access structure can be realized as a pure state scheme \cite[Theorem~8]{gottesman00},
it implies the following result,  stated explicitly due to its relevance for us.

\begin{corollary}\label{co:mat2qss}
An identically self-dual matroid $\mc{M}$ induces a pure state quantum secret sharing scheme. 
\end{corollary}
   
However, the preceding result does not give us  a method to construct a quantum secret sharing scheme
from the matroid, neither does it tell us if the scheme is ideal. 
The following theorem gives the general procedure to transform a representable identically self-dual matroid into a quantum secret sharing scheme. We denote a finite field with $q$ elements as $\F_q$. Following standard notation, we use $[n,k,d]_q$ to denote a 
classical code over $\F_q$ and $[[n,k,d]]_q$ to denote a quantum code over $\F_q$. 
 If $C$ is a code, we denote a generator matrix of $C$ by $G_C$. 
The code obtained by deleting the $i$th coordinate of $C$ is called a punctured code of $C$ and denoted as $\rho_i(C)$. 
Suppose we consider the subcode of  $C$ with the $i$th coordinate zero, then the code obtained by puncturing the $i$th coordinate
of the subcode is called a shortening of $C$ and denoted as $\sigma_i(C)$. We have the following useful relations between the punctured and shortened codes and their duals. 
\ben
\sigma_i(C) \subset \rho_i(C) \mbox{ and }  \sigma_i(C)^\perp = \rho_i(C^\perp). \label{eq:puncShortRelns}
\een
If $x\in \F_q^n$, then we denote the support of $x$ as $\supp(x)= \{ i\mid x_i\neq 0\}$. A codeword $x$ in $C$ is said to be
a  minimal support if there exists no nonzero codeword $y$ in $C$ such that $\supp(y) \subsetneq \supp(x)$. If in addition its leftmost 
nonzero component is 1, then it is said to be a minimal codeword.
Minimal codewords were introduced by Massey \cite{massey93}. They facilitate the study of classical secret sharing schemes, especially in characterizing the access structures.  

\begin{theorem}\label{th:isdMat2qss}
Let $\mc{M}(V,\mc{C})$ be an identically self-dual matroid representable over a finite field $\F_q$,  
where $V=\{0,1,\ldots, n-1,n \}$. Suppose that 
$C \subseteq \F_q^{n+1}$  such that the generator matrix of $C$ is a representation of $\mc{M}$. Let
\ben
G_C = \left[\ba{c|c}1 & g\\\bf{0} & G_{\sigma_0(C)} \ea  \right]  \mbox{ and }  
G_{\rho_0(C)}=\left[\ba{c}g\\ G_{\sigma_0(C)}  \ea\right].
\een
Then there exists an ideal pure state quantum secret sharing scheme $\Sigma$ on $P=\{ 1,\ldots, n\}$  whose 
access structure $\Gamma_0$ and  minimal
access structure   $\Gamma_{0,m}$, are defined by equations~\eqref{eq:mat2ss}~and~\eqref{eq:mat2ssMin}
respectively. 
The encoding for $\Sigma$ is determined by the stabilizer code with the stabilizer matrix given by
\ben
S=\left[\ba{c|c} G_{\sigma_0(C)}& \bf{0}\\\bf{0}&G_{\rho_0(C)^\perp} \ea\right].\label{eq:stabMatrix}
\een
The reconstruction procedure for an authorized set $A$  is the transformation on $S$ such that the encoded operators for
the transformed stabilizer code are $X_1=X\otimes I ^{\otimes^{n-1}}$ and $Z_1=Z\otimes I ^{\otimes^{n-1}}$.
\end{theorem}
\begin{proof}
The proof of this theorem is structured as follows. Since $\Sigma$ relies on the encoding of the stabilizer code derived from $S$,
we first show that $S$ defines a stabilizer code and identify certain properties of the codes $C$ and $C^\perp$ essential to 
recovering the secret. Then we show that if the secrets are encoded using the stabilizer encoding, then an element
 $A\in \Gamma_{0,m}$ does correspond to a minimal authorized set by explicitly reconstructing the secret with the shares in $A$
and proving that no proper subset of $A$ can reconstruct the secret. 

\smallskip
\noindent
Encoding the secret: 
We can easily  check that the matrix given in equation~(\ref{eq:stabMatrix}) does define a stabilizer code. We see that 
$\sigma_0(C)$ is an $[n,k-1,d]_q$ code, while $\rho_0(C)$ is an $[n,k,d-1]_q$ code with $\sigma_0(C)\subset \rho_0(C)$.
Therefore we have $\rho_0(C)^\perp\subset  \sigma_0(C)^\perp$ ensuring the orthogonality of 
$\sigma_0(C)$  and $\rho_0(C)^\perp$ in equation~(\ref{eq:stabMatrix}). 
The dimension of $S$ is given by $k-1+n-k=n-1$. Thus $S$ defines an $[[n,1,d']]_q$ quantum code, $Q$.

Since $\mc{M}(V,\mc{C})$ is an identically self-dual matroid,  both $C$ and $C^\perp$ represent $\mc{M}(V,\mc{C})$. 
Therefore, $g\neq0$, otherwise the zeroth column would be all zero in $C^\perp$ which would mean that
$\{0\}$ is a circuit, while from $C$, we would conclude that  $\{ 0\}$ is independent and not a circuit; a contradiction. 
Furthermore, without loss of generality we can choose $(1|g)$ to be a minimal codeword $c$ in $C$ \footnote{If $(1|g)$ is not minimal,
then there exists some codeword $(1|g')$ or $(0|a)$ such that its support is strictly contained in $\supp(1|g)$. If $\supp(0|a) \subset \supp(1|g)$, then we can find a codeword $(1|g')$, from a linear combination of $(1|g)$ and $(0|a)$, such that  $\supp(1|g') \subset \supp(1|g)$ and $\supp(0|a) \not\subset \supp(1|g')$.
In either case there is a codeword of the form $(1|g')$ whose support is strictly smaller than $\supp(1|g)$. If $(1|g')$ is minimal we
are done or we can repeat this process until we find one; the process will terminate in a finite number of steps as $n$ is finite.}.

As the support of a  minimal codeword in $C$ is a circuit of $\mc{M}(V,\mc{C})$, it follows that there exists a minimal codeword $c'$ in $C^\perp$ such that
$\supp(c')=\supp(c)$, in particular there exists a vector $(\beta | \beta g') \in C^\perp$ such that 
$\supp(\beta g')=\supp(g)$ for some $\beta \in \F_q^\times $ and $g'\in \rho_0(C^\perp)$.

\smallskip
\noindent
 The mapping for the secret sharing scheme is given as follows:
 \ben
 \ket{s} \mapsto \sum_{x\in \sigma_0(C)}\ket{s\cdot g + x }, \mbox{ where } s\in \F_q.
 \een
 Encoding of an arbitrary secret state follows by linearity of the encoding map.  The encoded $X$ operator for the quantum code is given by
 $\overline{X}=\otimes_{i=1}^n X^{g_i}$, or equivalently  $[\ba{c|c} g & \bf{0}\ea]$, its representation over $\F_q^{2n}$.

\smallskip
\noindent
Recovering the secret: Let $A \in \Gamma_{0,m}$, then $A\cup \{0 \} \in \mc{C}$ and there exists a minimal codeword $c' \in C^\perp$ such that 
$\supp(c') = A\cup\{ 0\}$.  Let $c'$ be a minimal codeword in $C^\perp$ such that $c_0'=1$. We know that there exists a $c\in C$ 
such  that $\supp(c) = \supp(c')$. We can choose $c_0=1$ since $C$ is a linear code.  
 Then,  we have $\rho_0(c)\not\in \sigma_0(C)$. Then both $\rho_0(c) $ and $g$ are in the 
 same coset of  $\sigma_0(C)$ in $\rho_0(C)$. This holds because the cosets of  $\sigma_0(C)$ in $\rho_0(C)$ 
 are in one to one correspondence
 with the cosets of $[{\bf 0}|\sigma_0(C) ]$ in $C$. The various coset representatives are given by $(\alpha | \alpha g)$, $\alpha \in \F_q$. Two coset representatives $r,r'$ represent the same coset if and only if $r_0= r_0'$. Therefore all the minimal codewords
 $c$, with $c_0=1$ are in the same coset as $(1|g)$. From this follows that 
  $\rho_0(c)$ is in the same coset as $(g)$. 
Therefore, the state $\ket{s}$ might as well be given by 
 \ben
 \ket{s} \mapsto \sum_{x\in \sigma_0(C)}\ket{s\cdot \rho_0(c) + x }.
 \een
 Denote the columns of $G_{\sigma_0(C)}$ by $s_i$, where $1\leq i\leq n$.
 Since $c' \in C^\perp$, we have 
 \be
G_C (c')^t = \left[ \ba{cccc} 1&c_1& c_2\dots & c_n \\ \bf{0} & s_1 & \dots & s_n \ea\right] \left[\ba{c}1\\c_1'\\ \vdots \\c_n' \ea\right]=\bf{0}.
 \ee
 The above equation can also be written as 
 \be
 \left[ \ba{cccc} c_1& c_2\dots & c_n \\  s_1 & \dots & s_n \ea\right] \left[\ba{c}-c_1'\\ \vdots \\-c_n' \ea\right]= \left[\ba{c} 1 \\\bf{0}\ea \right].
 \ee
 In other words, there exists a linear combination of the columns of $G_{\sigma_0(C)}$ such that 
 \ben
 \sum_{i\in \supp(\rho_0(c'))} c_i' s_i=0.\label{eq:recoveryLC}
 \een
Now let us rewrite the stabilizer and the encoded $X$ operator as follows.
 \be
\left[\ba{c}\overline{X} \\ \hline S \ea\right]
&=&\left[\ba{c|c}  \rho_0(c') & \bf{0} \\\hline G_{\sigma_0(C)}& \bf{0} \\{\bf{0}}& G_{\rho_0(C)^\perp} \ea\right]\\
&= &\left[\ba{cccccc|ccc}  c_1&\cdots &c_l&0&\cdots &0 & \multicolumn{3}{c}{\bf{0}} \\\hline 
s_1 & \multicolumn{4}{c}{\cdots} & s_n& \multicolumn{3}{c}{\bf{0}}\\ \multicolumn{6}{c|}{\bf{0}}&r_1&\cdots &r_n \ea\right],
 \ee
 where, without loss of generality, we can assume that $\rho_0(c')$  and therefore $\rho_0(c)$ have support in the first $l$ columns only, i.e., $c_i\neq0$ for $1\leq i\leq l$, and that  $c_i=0$ for $i>l$, where $1\leq l\leq n$. Note that $l\geq1$ because we must have
 $c\cdot c'=0$ and $l=0$ implies that $(1|0)\cdot (1|0)=0$ which is clearly not possible. 

 Let us transform the first column of $S$ as per equation~(\ref{eq:recoveryLC}) i.e., 
 $ s_1\mapsto \sum_{i\in \supp(\rho_0(c'))} c_i' s_i$.   Then we obtain  
 \be
\left[\ba{ccccccc|ccccccc}  1&c_2&\cdots &c_l&0&\cdots &0 & \multicolumn{7}{c}{\bf{0}} \\\hline 
\bf{0} &s_2 & \multicolumn{4}{c}{\cdots} & s_n& \multicolumn{7}{c}{\bf{0}}\\ 
\multicolumn{7}{c|}{\bf{0}}&r_1&\tilde{r}_2&\cdots &\tilde{r}_l&r_{l+1}& \cdots & r_n \ea\right].
 \ee
Therein, the columns $r_2$ to $r_l$ are transformed in the $Z$-part while only the first column is
 transformed in the $X$-part. For binary schemes this involves only CNOT gates, for nonbinary schemes, we have to
 use the generalized CNOT gates  \cite{grassl03}.  
 Now let us transform the encoded $X$ operator to the trivial operator given by $X_1$.
 This gives us 
 \be
 \left[\ba{ccccccc|ccccccc}  1&\multicolumn{6}{c|} {\bf{0}} & \multicolumn{7}{c}{\bf{0}} \\\hline 
\bf{0} &\tilde{s}_2 & \cdots & \tilde{s}_l & s_{l+1}&\cdots  & s_n& \multicolumn{7}{c}{\bf{0}}\\ 
\multicolumn{7}{c|}{\bf{0}}&\bf{0}&\tilde{r}_2&\cdots &\tilde{r}_l&r_{l+1}& \cdots & r_n \ea\right] 
\ee
which is in the form
\be
\left[\ba{cc|c}1&\bf{0} & \bf{0}\\\hline
\bf{0}&\tilde{S}_X& \bf{0} \\
\multicolumn{2}{c|}{\bf{0}}& \tilde{S}_Z\ea \right].
 \ee
 The column $\tilde{r}_{1}$ has to become zero because the stabilizer must commute with the encoded $X$ operator now
 given by $X_1$. The encoded secret has now been transformed as 
 \be
 \sum_{x\in \tilde{S}_X}\ket{s}\ket{x} = \ket{s}\sum_{x\in \tilde{S}_X}\ket{x} 
 \ee
 As can be seen above, the secret is completely disentangled from the rest of the qubits. Furthermore, in all these transformations
 we operated only on the qudits in the support of the minimal codeword. Thus the elements of 
 $\Gamma_{0,m}$ are authorized sets.
  
  \medskip
  \noindent
 Minimality of $\Gamma_{0,m}$: The sets in $\Gamma_{0,m}$
are minimal authorized sets because a proper subset of $\supp(c')$ is not in $C$ and all operators with
 their support in $\supp(c')$ must lie outside $C(S)$, the centralizer of $S$. 
 Such operators cannot extract any information about the encoded state
 as they correspond to detectable errors of the stabilizer code.

\smallskip
\noindent
Completeness of $\Gamma_0$:  
 By  Fact~\ref{fc:dualMat2ss},   the access structure $\Gamma_0$  is self-dual, therefore $|\Gamma_{0}|=2^{n-1}$. But the complement of everyone of these authorized sets is unauthorized set otherwise
 we would violate the no-cloning theorem. Together these sets exhaust all the possible subsets of $\{1,\ldots, n \}$.  Thus 
there are no more authorized sets outside the ones given in $\Gamma_0$ and 
the quantum access structure is completely defined by  $\Gamma_0$.

\smallskip
\noindent
Finally, the purity of $\Sigma$ follows from the explicit encoding procedure given. That $\Sigma$  is ideal follows from the fact that 
each share has the same dimension as the secret. 
 \end{proof}
We note that the choice of which element in the matroid is identified with the dealer is arbitrary. In
Theorem~\ref{th:isdMat2qss},  for simplicity we have assumed that the first element is the dealer. Further,
we need the representation of the matroid before we can use it to construct the scheme. 
It is still open if every identically self-dual matroid can be realized as a self-dual code, \cite{cramer08}. 
However, every self-dual code induces a 
identically self-dual matroid. Consequently, we have the following corollary that gives us many efficient pure state quantum secret
sharing schemes. 
\begin{corollary}
Let  $C \subseteq \F_q^{n}$  be an $[n+1,k,d]_q$ code such that $C^\perp= C$ with generator matrix $G_C$ given as
\ben
G_C = \left[\ba{c|c}1 & g\\\bf{0} & G_{\sigma_0(C)} \ea  \right] \mbox{ and } G_{\rho_0(C)}= \left[\ba{c}g\\G_{\sigma_0(C)} \ea  \right].
\een
Then there exists an ideal pure state quantum secret sharing scheme $\Sigma$ on $n$ parties whose minimal access structure 
is $$\Gamma_{0,m} = \left\{ \supp(c)\setminus \{ 0\} \bigg| \ba{l} c \mbox{ is a minimal codeword} 
\\\mbox{of  } C \mbox{ with  } c_0 =1\ea\right\}.$$
The encoding for $\Sigma$ is determined by the stabilizer code with the stabilizer matrix given by
\ben
S=\left[\ba{c|c} G_{\sigma_0(C)}& \bf{0}\\\bf{0}&G_{\rho_0(C)^\perp} \ea\right].
\een
The encoding  for a state $\ket{s}$, where $s\in \F_q$,   is given as
\ben
 \ket{s} \mapsto \sum_{x\in \sigma_0(C)} \ket{s\cdot g+ x}.
\een
The reconstruction procedure for an authorized set $A$ of $\Sigma$ is the transformation on $S$ such that the encoded operators for
the transformed stabilizer are $X_1$ and $Z_1$.
\end{corollary}
\noindent
\subsection{An example}
We now give an example to illustrate the construction. 
Let us consider the extended Hamming code given by the following generator matrix. 
\be
G_C=\left[\ba{c|ccccccc}1&1&1&1&1&1&1&1\\\hline 0&1&0&0&0&1&1&1\\0&0&1&0&1&0&1&1\\0&0&0&1&1&1&1&0
\ea\right]
\ee
We can check that $C$ is self-dual. The punctured code $\rho_0(C)$ and the shortened code $\sigma_0(C)$  are given by the 
following generator matrices. (For clarity we show only nonzero entries).
\be
G_{\rho_0(C)}&=&\left[\ba{ccccccc}1&1&1&1&1&1&1\\\hline 1&0&0&0&1&1&1\\0&1&0&1&0&1&1\\0&0&1&1&1&1&0 \ea\right]
\\ G_{\sigma_0(C)}&=&\left[\ba{ccccccc} 1&0&0&0&1&1&1\\0&1&0&1&0&1&1\\0&0&1&1&1&1&0 \ea\right]
\ee
Now let us form a (CSS) stabilizer code with stabilizer matrix as follows. 
\be
S &=& \left[ \ba{c|c}G_{\sigma_0(C)}&\bf{0} \\ \hline \bf{0} & \rho_{0(C)^\perp}\ea \right] \\
&=& 
\bmt{ccccccc | ccccccc}  1&0&0&0&1&1&1&\multicolumn{7}{c}{\multirow{3}{*}{\bf{0}}}\\0&1&0&1&0&1&1\\0&0&1&1&1&1&0& \\\hline
\multicolumn{7}{c|}{\multirow{3}{*}{\bf{0}}}&1&0&0&0&1&1&1\\ \multicolumn{7}{c|}{}&0&1&0&1&0&1&1\\ \multicolumn{7}{c|}{}&0&0&1&1&1&1& 0
 \emt
\ee
For this stabilizer code the encoding for $\ket{0}$ and $\ket{1}$ is given as follows:
\be
\ket{0}&\mapsto&\ket{0000000}+\ket{1000111}+\ket{0101011}+\ket{0011110}\\
&+&\ket{1101100}+\ket{1011001}+\ket{0110101}+\ket{1110010}\\ 
\ket{1}&\mapsto&\ket{1111111}+\ket{0111000}+\ket{1010100}+\ket{1100001}\\
&+&\ket{0010011}+\ket{0100110}+\ket{1001010}+\ket{0001101}
\ee
Observe that $\ket{s}\mapsto \sum_{c\in \sigma_0(C)} \ket{s\cdot \overline{X}+c}$, where $\overline{X}$ is in $\sigma_0(C)^\perp\setminus \sigma_0(C)$.
Now consider a minimal codeword $c$ in $C^\perp$ such that $c_0=1$. One such codeword is $c=(1,1,1,0,0,0,0,1)$. Observe that 
$(1,1,0,0,0,0,1|\bf{0})$ i.e. $X_1X_2X_7$ is an encoded operator for the stabilizer code.
The support of $c$ is given by $\{ 0,1,2,7\}$.  The linear combination of the columns in the support of $\rho_0(c)$ gives us 
$(1,0,0,0)^t$. If we computed the linear combination of the columns $1,2,7$ into the first column we get the following action on 
stabilizer of the code. 
\be
S  \mapsto
\bmt{ccccccc | ccccccc}  0&0&0&0&1&1&1&\multicolumn{7}{c}{\multirow{3}{*}{\bf{0}}}\\0&1&0&1&0&1&1\\0&0&1&1&1&1&0& \\\hline
\multicolumn{7}{c|}{\multirow{3}{*}{\bf{0}}}&1&1&0&0&1&1&0 \\ \multicolumn{7}{c|}{}&0&1&0&1&0&1&1\\ \multicolumn{7}{c|}{}&0&0&1&1&1&1& 0
 \emt
\ee
The encoded operator $(1,1,0,0,0,0,1| \bf{0} )$ maps to $  (1,1,0,0,0,0,1|\bf{0}) $. If we now transform the encoded operator to 
$(1,0,0,0,0,0,0|\bf{0})$ the stabilizer gets further transformed as 
\be
S  \mapsto
\bmt{ccccccc | ccccccc}  0&0&0&0&1&1&1&\multicolumn{7}{c}{\multirow{3}{*}{\bf{0}}}\\0&1&0&1&0&1&1\\0&0&1&1&1&1&0& \\\hline
\multicolumn{7}{c|}{\multirow{3}{*}{\bf{0}}}&0&1&0&0&1&1&0 \\ \multicolumn{7}{c|}{}&0&1&0&1&0&1&1\\ \multicolumn{7}{c|}{}&0&0&1&1&1&1& 0
 \emt
\ee
Observe that this time only the $Z$ part of the stabilizer is transformed. Now the encoded secret is in the state  
\be
&&\ket{s}\left( \ket{000000} +\ket{000111} + \ket{101011} + \ket{011110}\right.\\
&&+\left.\ket{101100}+\ket{011001}+\ket{110101}+ \ket{110010} \right)
\ee
The secret is completely disentangled from the rest of the shares. Therefore, $\supp(c)\setminus \{0\}$ forms an authorized set.
The rest of the shares cannot reconstruct or extract any information from their shares because of the no-cloning theorem.
Similarly any minimal codeword in $C^\perp$ with $c_0=1$ defines an authorized set for the scheme. Suppose that $c$ is a minimal
codeword with $c_0=0$, then it must be in $\sigma_0(C)$ and any other vector whose support is the same must be in 
$S$ or outside  $C(S)$, the centralizer of  $S$. No such operator can reveal any information about the encoded secret since they are
detectable errors of the stabilizer code and by definition detectable errors reveal nothing about the encoded information.

\subsection{Discussion}

The results in this section have important benefits. Quantum secret sharing schemes for general access structures were proposed
by Gottesman \cite{gottesman00} and Smith \cite{smith00}, based on monotone span programs. These constructions are not optimal in general.
Our method gives optimal schemes with information rate one. However, not every ideal quantum secret sharing scheme
can be derived by Theorem~\ref{th:isdMat2qss}.  For instance, 
the $((3,5))$ threshold scheme  can be realized using the $[[5,1,3]]$ quantum code, but it cannot be realized by the method
proposed.  Furthermore, the access structure of the $((3,5))$ scheme induces a matroid. 
It would be worth investigating to find out how such quantum schemes can derived from matroids. Another interesting question
would be to derive ideal quantum secret sharing schemes from non-representable ISD matroids. 

\subsection*{Acknowledgment}
We thank Hoi-Kwong Lo for pointing out an inaccuracy in an earlier version of the paper. 
Part of this work was presented at the Workshop on Applications of Matroid Theory and Combinatorial Optimization to Information and Coding Theory,  Banff International Research Station, Banff, 2009.
This research was supported by NSERC, CIFAR and MITACS.


\end{document}